\documentclass[a4paper,UKenglish,cleveref, autoref, thm-restate]{lipics-v2021}

\usepackage[utf8]{inputenc} % allow utf-8 input
\usepackage[T1]{fontenc}    % use 8-bit T1 fonts
\usepackage{hyperref}       % hyperlinks
\usepackage{url}            % simple URL typesetting
\usepackage{booktabs}       % professional-quality tables
\usepackage{amsfonts}       % blackboard math symbols
\usepackage{nicefrac}       % compact symbols for 1/2, etc.
\usepackage{microtype}      % microtypography
\usepackage[dvipsnames]{xcolor}         % colors
\usepackage{tikz}
\usetikzlibrary{positioning}

\nolinenumbers

\providecommand{\citep}[1]{\cite{#1}}
\providecommand{\va}[1]{}

\newcommand{\hlcolor}{Yellow!35}
\newcommand{\hlcolorTwo}{LimeGreen!35}
\makeatletter
\newenvironment{btHighlight}[1][]
{\begingroup\tikzset{bt@Highlight@par/.style={#1}}\begin{lrbox}{\@tempboxa}}
{\end{lrbox}\bt@HL@box[bt@Highlight@par]{\@tempboxa}\endgroup}

\newcommand\btHL[1][]{%
  \begin{btHighlight}[#1]\bgroup\aftergroup\bt@HL@endenv%
}
\def\bt@HL@endenv{%
  \end{btHighlight}%   
  \egroup
}
\newcommand{\bt@HL@box}[2][]{%
  \tikz[#1]{%
    \pgfpathrectangle{\pgfpoint{1pt}{0pt}}{\pgfpoint{\wd #2}{\ht #2}}%
    \pgfusepath{use as bounding box}%
    \node[anchor=base west, fill=\hlcolor,outer sep=0pt,inner xsep=1pt, inner ysep=0pt, rounded corners=2pt, minimum height=\ht\strutbox+2pt,#1]{\raisebox{1pt}{\strut}\strut\usebox{#2}};
  }%
}

\newenvironment{btHighlightTwo}[1][]
{\begingroup\tikzset{bt@HighlightTwo@par/.style={#1}}\begin{lrbox}{\@tempboxa}}
{\end{lrbox}\bt@HLTwo@box[bt@HighlightTwo@par]{\@tempboxa}\endgroup}

\newcommand\btHLTwo[1][]{%
  \begin{btHighlightTwo}[#1]\bgroup\aftergroup\bt@HLTwo@endenv%
}
\def\bt@HLTwo@endenv{%
  \end{btHighlightTwo}%   
  \egroup
}
\newcommand{\bt@HLTwo@box}[2][]{%
  \tikz[#1]{%
    \pgfpathrectangle{\pgfpoint{1pt}{0pt}}{\pgfpoint{\wd #2}{\ht #2}}%
    \pgfusepath{use as bounding box}%
    \node[anchor=base west, fill=\hlcolorTwo,outer sep=0pt,inner xsep=1pt, inner ysep=0pt, rounded corners=2pt, minimum height=\ht\strutbox+2pt,#1]{\raisebox{1pt}{\strut}\strut\usebox{#2}};
  }%
}

\usepackage{multicol}
\usepackage{listings}
\usepackage{amsthm,amsmath}

\title{Concurrent Splay-Based Tree}

\author{Vitaly Aksenov}{ITMO University, Russia}{aksenov@itmo.ru}{}{}

\author{Rene van Bevern}{Independent Researcher}{m5bere2@proton.me}{}{}

\author{Artem Shilkin}{ITMO University, Russia}{bpropeller.03@gmail.com}{}{}

 \authorrunning{V. Aksenov, R. van Bevern, and A. Shilkin}

\Copyright{Vitaly Aksenov, Rene van Bevern, Artem Shilkin} %TODO mandatory, please use full first names. LIPIcs license is "CC-BY";  http://creativecommons.org/licenses/by/3.0/

 %TODO mandatory: Please choose ACM 2012 classifications from https://dl.acm.org/ccs/ccs_flat.cfm 


\EventEditors{John Q. Open and Joan R. Access}
\EventNoEds{2}
\EventLongTitle{42nd Conference on Very Important Topics (CVIT 2016)}
\EventShortTitle{CVIT 2016}
\EventAcronym{CVIT}
\EventYear{2016}
\EventDate{December 24--27, 2016}
\EventLocation{Little Whinging, United Kingdom}
\EventLogo{}
\SeriesVolume{42}
\ArticleNo{23}
\begin{document}
\maketitle

\begin{abstract}
Most work on efficient concurrent ordered indices, such as concurrent binary search trees, B-trees, skip lists, etc., has focused on data structures that provide good \emph{worst-case} guarantees. In real workloads, objects are often accessed at different rates, since access distributions may be non-uniform. Many efficient distribution-adaptive data structures exist in the sequential case; however, they are often complicated to make efficient in the concurrent case.

The most prominent distribution-adaptive data structure is Splay Tree. Its most important advantage is that it does not store any balancing information and provides a reasonable performance improvement on extremely skewed workloads, such as Zipfian workloads. This paper proposes a splay-like rotation design for concurrent binary search trees. Instead of moving an accessed node to the root, rotations use two depth thresholds that are based on the static-optimality complexity computed from the number of accesses to the node: a node is rotated only when it is substantially deeper than the upper threshold, and rotations of the node stop before reaching the lower threshold. This design aims to preserve the main practical benefit of splaying on skewed workloads while reducing contention near the root.

We present two variants of the rotation design: one using an exact 64-bit access counter per node and one using a 6-bit approximate counter. We prove static optimality for the corresponding sequential read-only tree and evaluate both rotation designs by implementing them on top of the concurrent AVL tree of Bronson et al. Our experiments show that the approach can improve throughput on several skewed workloads.

\vspace{0.5cm}

% \noindent\textbf{AI Disclosure: We used Codex with ChatGPT-5.5 to assist with: 1)~grammar and consistency of the notation; and 2)~writing the background Section 2.}
\end{abstract}

% \newpage

\section{Introduction}
\vspace{-0.3cm}

There has been a significant effort to design concurrent data structures from sequential variants, e.g., hash tables~\citep{Michael, HSBook}, skip lists~\citep{FraserThesis, LazySkiplist, DougSkiplist}, and search trees~\citep{Natarajan14, bronson2010practical}. However, most of these works have focused on data structures with optimal worst-case guarantees. It is known that the worst-case complexity of point operations in any data structure built using only comparisons is logarithmic in the size of the data structure. 
One of the ways to escape this lower bound is to adapt to the workload, i.e., requests.
% We are aware of two ways to escape this lower bound: adapt to the stored data or adapt to requests. In the first case, one can utilise patterns in the underlying data. For example, if all keys in a data structure are even numbers from $2$ to $2n$, then one can obtain the position of a key using just linear interpolation. Data structures that use patterns in data are known as learned indices~\cite{al2025survey}. Unfortunately, these data structures usually require more than a simple comparison, so they do not work well on all types of keys, e.g., strings. Due to this limitation, in this work, we focus on the second approach for escaping the worst case~--- adapting to requests.

Fortunately, in many real workloads, the access rates for keys are not uniform. This fact is well known and is modelled in several industrial benchmarks, such as YCSB~\citep{cooper2010benchmarking} and TPC-C~\citep{TPCC}, where the generated access distributions are heavy-tailed, e.g., following a Zipfian distribution~\citep{cooper2010benchmarking}. There is a large line of work devoted to adaptive data structures in the sequential case; see, e.g.,~\citep{knuth1997art} and references therein, with the most renowned one being Splay Tree~\cite{sleator1985self}. One can consider it as one of the simplest yet efficient data structures using no balancing information and working exceptionally well on Zipfian workloads, with the most accessed element always very close to the root. Unfortunately, Splay Tree is not easy to make concurrent efficiently due to its rebalancing procedure, which rotates the accessed node to the root. This has an effect on its concurrent implementation: multiple concurrent operations will contend for the root, leading to a large bottleneck.

Because of that issue, we are unaware of any work on making the original Splay Tree concurrent without changes. Instead, researchers have proposed two adaptive data structures based on specialized rotations: CBTree~\citep{CBTree} and Splay-List~\cite{aksenov2020splay}. CBTree uses a rebalancing procedure similar to those of both AVL and Splay Trees. Intuitively, it tries to maintain the following invariant for each node: the ratio between the numbers of requests to the left and right subtrees should not exceed some predefined constant. Compared to Splay Tree, which does not store any additional information, CBTree requires storing three integers per node: the number of accesses to it and the number of accesses to both subtrees (this can be improved to two integers by storing just the size of its own subtree). In addition to the memory overhead, it appeared to scale a bit worse than the second known concurrent adaptive data structure, Splay-List. Splay-List is based on a Skip List~\cite{Pugh} with a rebalancing procedure that uses the number of accesses. Unfortunately, better scaling comes with even worse memory utilization: Splay-List requires much more memory than CBTree, since it stores the number of accesses per node plus the number of accesses per ``subtree''.

In this work, we try to come up with data structures based on Splay Tree that work better and use less additional memory than the presented counterparts. For that, we propose two concurrent variations of Splay Tree that use the same rotations as the vanilla one while requiring only one integer per node: the number of accesses. The idea is quite simple. The usual goal of an adaptive data structure is to provide static optimality for accesses, recalled in Section~\ref{sec:background}. Splay Tree satisfies this property, but, as we explained, its straightforward implementation creates a bottleneck close to the root. Thus, instead of splaying to the root, we rotate the node only up to depth $A \cdot \log (m / ac(x))$. That change is not enough, since nodes would still be splayed on each access. To reduce the number of rotations, we rotate the node only if its current depth exceeds $B \cdot \log (m / ac(x))$. Finally, we initiate the rotation only with some probability, which was experimentally shown to be helpful for Splay Tree~\cite{albers2002randomized} while maintaining the complexity.

Our first version of the rotation design stores the number of accesses in each node using 64 bits. Then, we propose a second version based on an approximate counter by Morris~\cite{morris1978counting}. Instead of the number of accesses, it stores roughly its logarithm, and on an access it increments its value, $r$, with probability $\frac{1}{2^r}$.

We analyze the first proposed rotation design in a sequential read-only setting and show that it preserves static optimality. 
We then evaluate our design in a concurrent setting by integrating it into the AVL tree of Bronson et al.~\cite{bronson2010practical}. The implementation is compared with the original AVL tree and with CBTree. The experiments show that our rotation design improves the performance on several skewed workloads.

The paper has the following structure. In Section~\ref{sec:background}, we recall the necessary background on Splay Tree and static optimality. In Section~\ref{sec:sequential}, we propose the design of our rotation procedure. In Section~\ref{sec:proof}, we prove that this design leads to the static-optimality property. In Section~\ref{sec:experiments}, we present the evaluation of our data structure.

\vspace{-0.3cm}
\section{Background}
\label{sec:background}
\vspace{-0.2cm}

Splay tree~\cite{sleator1985self} is a binary search tree that does not store balancing information. An access (or \texttt{get} operation) to a key $x$ traverses the tree as usual. If $x$ is found in a node $u$, the tree then \emph{splays} $u$: it repeatedly applies local rotations that move $u$ toward the root. 
%The search cost is the length of the path to $u$, and the rotations are intended to make future accesses to $x$ and to nearby keys cheaper.

The splay step uses three types of rotations, shown in Figure~\ref{fig:splay-rotations}. In the \emph{zig} case, $u$ has no grandparent, so one ordinary rotation makes $u$ the root. In the \emph{zig-zig} case, $u$ and its parent are both left children or both right children; the algorithm rotates the parent and then $u$. In the \emph{zig-zag} case, $u$ is a left child and its parent is a right child, or symmetrically; the algorithm rotates $u$ twice. The zig-zig and zig-zag cases decrease the depth of $u$ by two, while the final zig decreases it by one.

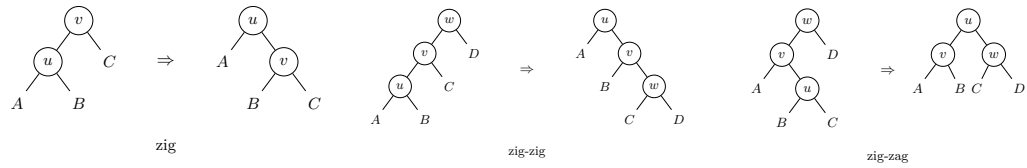
\begin{figure}[t]
\centering
\tikzset{
    splaynode/.style={circle,draw,minimum size=5.5mm,inner sep=0pt,font=\small},
    subtree/.style={font=\small},
    treeedge/.style={-}
}
\begin{minipage}[t]{0.31\textwidth}
\centering
\resizebox{\linewidth}{!}{%
\begin{tikzpicture}[level distance=8mm,sibling distance=12mm,baseline]
    \node[splaynode] (p) {$v$}
        child {node[splaynode] (u) {$u$}
            child {node[subtree] {$A$}}
            child {node[subtree] {$B$}}}
        child {node[subtree] {$C$}};
    \node at (1.7,-0.8) {$\Rightarrow$};
    \begin{scope}[xshift=3.4cm]
    \node[splaynode] {$u$}
        child {node[subtree] {$A$}}
        child {node[splaynode] {$v$}
            child {node[subtree] {$B$}}
            child {node[subtree] {$C$}}};
    \end{scope}
    \node at (1.7,-2.45) {\small zig};
\end{tikzpicture}
}
\end{minipage}\hfill
\begin{minipage}[t]{0.31\textwidth}
\centering
\resizebox{\linewidth}{!}{%
\begin{tikzpicture}[level distance=8mm,sibling distance=12mm,baseline]
    \node[splaynode] {$w$}
        child {node[splaynode] {$v$}
            child {node[splaynode] {$u$}
                child {node[subtree] {$A$}}
                child {node[subtree] {$B$}}}
            child {node[subtree] {$C$}}}
        child {node[subtree] {$D$}};
    \node at (1.9,-1.2) {$\Rightarrow$};
    \begin{scope}[xshift=3.8cm]
    \node[splaynode] {$u$}
        child {node[subtree] {$A$}}
        child {node[splaynode] {$v$}
            child {node[subtree] {$B$}}
            child {node[splaynode] {$w$}
                child {node[subtree] {$C$}}
                child {node[subtree] {$D$}}}};
    \end{scope}
    \node at (1.9,-3.2) {\small zig-zig};
\end{tikzpicture}
}
\end{minipage}\hfill
\begin{minipage}[t]{0.31\textwidth}
\centering
\resizebox{0.88\linewidth}{!}{%
\begin{tikzpicture}[level distance=8mm,sibling distance=12mm,baseline]
    \node[splaynode] {$w$}
        child {node[splaynode] {$v$}
            child {node[subtree] {$A$}}
            child {node[splaynode] {$u$}
                child {node[subtree] {$B$}}
                child {node[subtree] {$C$}}}}
        child {node[subtree] {$D$}};
    \node at (1.9,-1.2) {$\Rightarrow$};
    \begin{scope}[xshift=3.8cm]
    \node[splaynode] {$u$}
        child {node[splaynode] {$v$}
            child {node[subtree] {$A$}}
            child {node[subtree,xshift=-2mm] {$B$}}}
        child {node[splaynode] {$w$}
            child {node[subtree,xshift=2mm] {$C$}}
            child {node[subtree] {$D$}}};
    \end{scope}
    \node at (1.9,-3.2) {\small zig-zag};
\end{tikzpicture}
}
\end{minipage}
\vspace{-0.8em}
\caption{The left-child versions of Splay Tree rotations. The right-child versions are symmetric.}
\label{fig:splay-rotations}
\vspace{-1.5em}
\end{figure}

\begin{definition}[Static optimality]
Consider a sequence of $m$ successful access operations. Let $ac(x)$ be the number of accesses to key $x$. A search tree is \emph{statically optimal} if its total access cost is $O\left( \sum_x ac(x) \cdot \log (m / ac(x)) \right)$, or, equivalently, if the amortized cost of an access to key $x$ is $O(\log (m / ac(x)))$.
% This definition concerns access operations only; it does not include update operations.
\end{definition}

\vspace{-0.3cm}
\section{Rotation Procedure}
\label{sec:sequential}
\vspace{-0.2cm}

The simplest approach to avoid the bottleneck of the Splay Tree in the concurrent setting is to splay the accessed node only up to some fixed depth and perform rotations only up to some number of conflicts (i.e., when one rotation is overtaken by another concurrent rotation). Unfortunately, our experiments showed that such an approach does not work well due to at least two issues. First, we were unable to find universal constants that work well for every tree size and every number of working threads. Such tuning, of course, is undesirable. Second, it provides quite a large constant overhead on skewed workloads due to the inability to move nodes up the depth threshold. For example, on a Zipfian workload, the most commonly accessed element may be lower than the threshold, and thus will never be moved past the depth threshold, leading to a traversal overhead. So, the choice of the threshold is complicated: it should be large enough to reduce the bottleneck, but at the same time it should be small.

\vspace{-0.3cm}
\subsection{Version with exact counters}
\vspace{-0.2cm}

Since the trivial approach does not seem to work, we had to find another way to bound the depth up to which we rotate the node. We decided to choose this bound by trying to satisfy the static-optimality property: rotate up to the level $A \cdot \log (m / ac(x))$, where $m$ is the total number of accesses and $ac(x)$ is the number of accesses to the key. Unfortunately, with only that bound, the nodes lower than it will always be splayed leading to a large number of rotations. Thus, we decided to allow some slack and rotate only if the depth of the node is much larger than the static-optimality property requires, i.e., exceeds $B \cdot \log (m / ac(x))$.

% Our design is based on the rotations of Splay Tree: zig-zig, zig-zag, and zig. During the splay, the first two are usually applied and reduce the depth by two, while the last one happens only at the last moment to make the accessed node the root. We hide the choice of the corresponding rotation inside the rotate function.

\vspace{-1em}
\begin{center}
\begin{minipage}{\textwidth}
\begin{lstlisting}
void get(x):
  depth = 0 #\label{line:start-traversal}#
  node = root
  while node != null:
    if node.key = x:
      break
    depth += 1
    if node.key < x:
      node = node.right
    else:
      node = node.left #\label{line:end-traversal}#

  if node == null: #\label{line:not-found-1}#
    return null  #\label{line:not-found-2}#

  node.counter += 1 #\label{line:calculate:1}#
  global_counter += 1
  target = log(global_counter / node.counter) #\label{line:calculate:2}#

  // We allow the slack
  if depth < B * target: #\label{line:slack:1}#
    return node.value #\label{line:slack:2}#

  if rand() > PTHRESHOLD: #\label{line:random-rebalance}#
    return node.value #\label{line:random-stop}#

  // Rotate only up to some level
  while depth > A * target: #\label{line:rotation:1}#
    rotate(node)
    depth -= 2 #\label{line:rotation:2}#
  // zig-zig and zig-zag reduce depth
  // by 2. zig happens only at the root.
\end{lstlisting}
\vspace{-2.0em}
\captionof{lstlisting}{The code of the access (\texttt{get}) operation.}
\label{alg:original}
\vspace{-0.2cm}
\end{minipage}
\end{center}

First, in Lines~\ref{line:start-traversal}-\ref{line:end-traversal}, we traverse the tree to find a node with the requested key and simultaneously calculate the depth of the target node. If the key is not found, we return \texttt{null} (Lines~\ref{line:not-found-1}-\ref{line:not-found-2}). Otherwise, we increment the global counter and the key counter, and calculate the target value of the static-optimal complexity (Lines~\ref{line:calculate:1}-\ref{line:calculate:2}). If the node has depth less than some constant $B$ times the target complexity, we end the operation (Lines~\ref{line:slack:1}-\ref{line:slack:2}). Then, we generate a random number to decide whether we need to perform rotations (Line~\ref{line:random-rebalance}). If we decide not to, we return the value (Line~\ref{line:random-stop}). Otherwise, we rotate the node toward the root until its depth becomes less than $A$ times the target (Lines~\ref{line:rotation:1}-\ref{line:rotation:2}). Note that during splaying, usually zig-zig and zig-zag rotations are used, reducing the depth by $2$. The zig rotation can happen only at the very top.

The concurrent implementation follows the exact same procedure. We have three notes about it: 1)~increments of the counters are implemented using the fetch-and-add operation; 2)~the increments to the global counter may lead to contention, which will be fixed in the second design; and 3)~the depth may not be calculated precisely and the node may not be rotated to the expected place due to concurrent rotations.

\vspace{-0.3cm}
\subsection{Static-optimality proof}
\label{sec:proof}
\vspace{-0.2cm}

In this subsection, we prove that the sequential tree with our rotation design provides static optimality. To prove it, we require a tree to be initially filled with all keys, and to serve only access, i.e., \texttt{get}, requests.

Let $w(u)$ be the number of accesses to the node at the end, let $\sigma(u)$ be the sum of weights of all nodes in the subtree of $u$, and let the rank be $\rho(u) = \log \sigma(u)$.
Suppose that we decide to make a rotation with probability $p$ (\texttt{PTHRESHOLD}).
Thus, we will introduce a potential function $\Phi = (1/p + d) \cdot \sum \rho(u)$ where $d$ is the cost of one rotation.

\vspace{-0.2cm}
\begin{lemma}
\label{lem:height}
The expected amortized time incurred by the $i$-th operation is $O(\log(AC_i / ac_i(t)) + \log(m / ac_m(t)))$, where $t$ is the node found by the operation, $ac_i(t)$ is the total number of accesses to $t$, $AC_i$ is the total number of accesses before the $i$-th access, and $m$ is the total number of requests.
\end{lemma}
\begin{proof}
Let us start with the simplest case. If the depth of the node does not exceed $B \cdot \log AC_i / ac_i(t)$ by the condition in Line~\ref{line:slack:1}, the statement follows directly.

Now, we consider the case when the depth exceeds that bound.
Let us calculate the expected cost of the change of the potential and the cost of the operation.

We split an access operation into two parts. In the first part, we traverse to node $s$ at depth $A \cdot \log(AC_i / ac_i(t))$ (adding that cost). Then, we traverse to the target node $t$ and splay it up to $s$ with some probability. For simplicity, we assume that $t$ takes the place of $s$ and not a slightly higher position due to the parity of the depth.

For the second part, we repeat a standard potential-based proof for Splay Tree from \cite{albers2002randomized} and get the bound $\log (m / ac_m(t))$. We provide the full version in the Appendix.
\end{proof}

The main result follows from this lemma almost straightforwardly~--- we just need to replace $AC_i / ac_i(t)$ by $m / ac_m(t)$. The full proof is in the Appendix.

\begin{theorem}
A tree with such rotations serves accesses in $O(1/p \cdot \sum_x ac(x) \cdot \log (m / ac(x)))$ amortized time, where $m$ is the total number of accesses and $ac(x)$ is the total number of accesses to $x$.
\end{theorem}

\vspace{-0.3cm}
\subsection{Version with approximate counters}
\vspace{-0.2cm}

The second version of our rotation design is not much different from the one above. Instead of exact counters, we use the approximate counter from~\cite{morris1978counting} for global and node counters: on an access, we increment a counter with value $r$ with probability $\frac{1}{2^r}$ and return $2^r$. The Morris counter gives a compact estimate whose expectation is within a constant-factor scale of the true count. This counter needs just $6$ bits to approximate a $64$-bit counter.
% We use this simple approach in our code: the code from Listing~\ref{alg:approximate} should replace Lines~\ref{line:calculate:1}-\ref{line:calculate:2} of Listing~\ref{alg:original}.
The pseudocode is shown in Listing~\ref{alg:approximate}.

Unfortunately, the basic Morris counter does not provide direct static-optimality bounds due to nontrivial variance. Instead, one can use the improvement presented in~\cite{nelson2022optimal}. Nevertheless, we use the basic approach in the implementation.

\vspace{-0.3cm}
\section{Experiments}
\label{sec:experiments}
\vspace{-0.2cm}

We implemented our rotation design on top of the concurrent AVL tree by Bronson et al.~\cite{bronson2010practical} and obtained two data structures: Splay-like and Approximate-Splay-like. We compare their performance against the AVL tree by Bronson et al.~\cite{bronson2010practical} and its adaptive version, CBTree~\cite{CBTree}. We implemented the code in Java. We did not compare with Splay-List~\cite{aksenov2020splay} since it was written in C++.

As noted, the issue with the simplest concurrent Splay Tree is fixing the best constants. Fortunately, we were able to find good generic parameters for our data structures. For the Splay-like tree, we set the probability of splaying to $\frac{1}{20 \cdot T}$, where $T$ is the number of threads, the upper bound constant is $B = 2.5$, and the lower bound constant is $A = 0.7$. For the Approximate-Splay-like tree, we set the probability to $1$, $B = 2$, and $A = 0.5$. We also had to update CBTree to rotate with some probability in order to improve its performance. The probability was set to $\frac{1}{10 \cdot T}$.

For the workloads, we chose ones similar to the workloads presented in the Splay-List paper~\cite{aksenov2020splay}. First, we fix the range to $10^6$ elements. Then, we fix five distributions: 1)~\texttt{uniform}~--- a key is chosen uniformly; 2)~\texttt{zipfian}~--- a key is chosen from a Zipfian distribution with $\alpha = 1$ over a preliminarily shuffled set; 3)~\texttt{99/1}, \texttt{95/5}, and \texttt{90/10}~--- \texttt{x/y} means that we choose a set of $x\%$ random elements from the range, and the key is chosen with probability $y\%$ from this set and otherwise from the rest. We also choose two types of workloads. In read-only workloads, the structure is pre-filled with all the keys, and \texttt{get} operations choose the key from the distribution. In update workloads, the data structure is pre-filled with a random half of the range, and the operation is \texttt{get} with probability $80\%$, with the key from the distribution, or \texttt{insert}/\texttt{remove} with probability $10\%$ each, with the key taken uniformly from the range.

The experiments were run three times on a machine with four x86 chips with 16 cores each (64 cores in total) and 256 GB of RAM, with each run consisting of 20 seconds of warmup and 20 seconds of measured execution. The throughput shown in the plots was averaged. The code was written in Java and compiled with OpenJDK 21.0.11.

Figure~\ref{fig:read-throughput} shows the results of the experiments on the read-only workloads. As one can see, the original AVL outperforms the other data structures on lower-skew workloads, i.e., \texttt{uniform} and \texttt{90/10}, and on Zipfian workloads due to its static structure and, thus, better cache performance. On higher regular skew, \texttt{95/5} and \texttt{99/1}, our trees work similarly to or better than the original AVL. Please note that our approximate-counter version works similarly to the exact-counter version. CBTree performs worse than expected in our implementation, even though we followed the pseudocode from the paper~\cite{CBTree} and tuned for the best parameters.

Figure~\ref{fig:update-throughput} shows the results of the experiments on the update workloads. This time, AVL has worse cache usage, and our data structures outperform it on almost all workloads while working a little worse on the uniform workload. Our approximate-counter version works slightly better than the exact-counter version.

\begin{figure}[!t]
\centering
\begin{subfigure}[t]{0.32\textwidth}
    \centering
    \includegraphics[width=\linewidth]{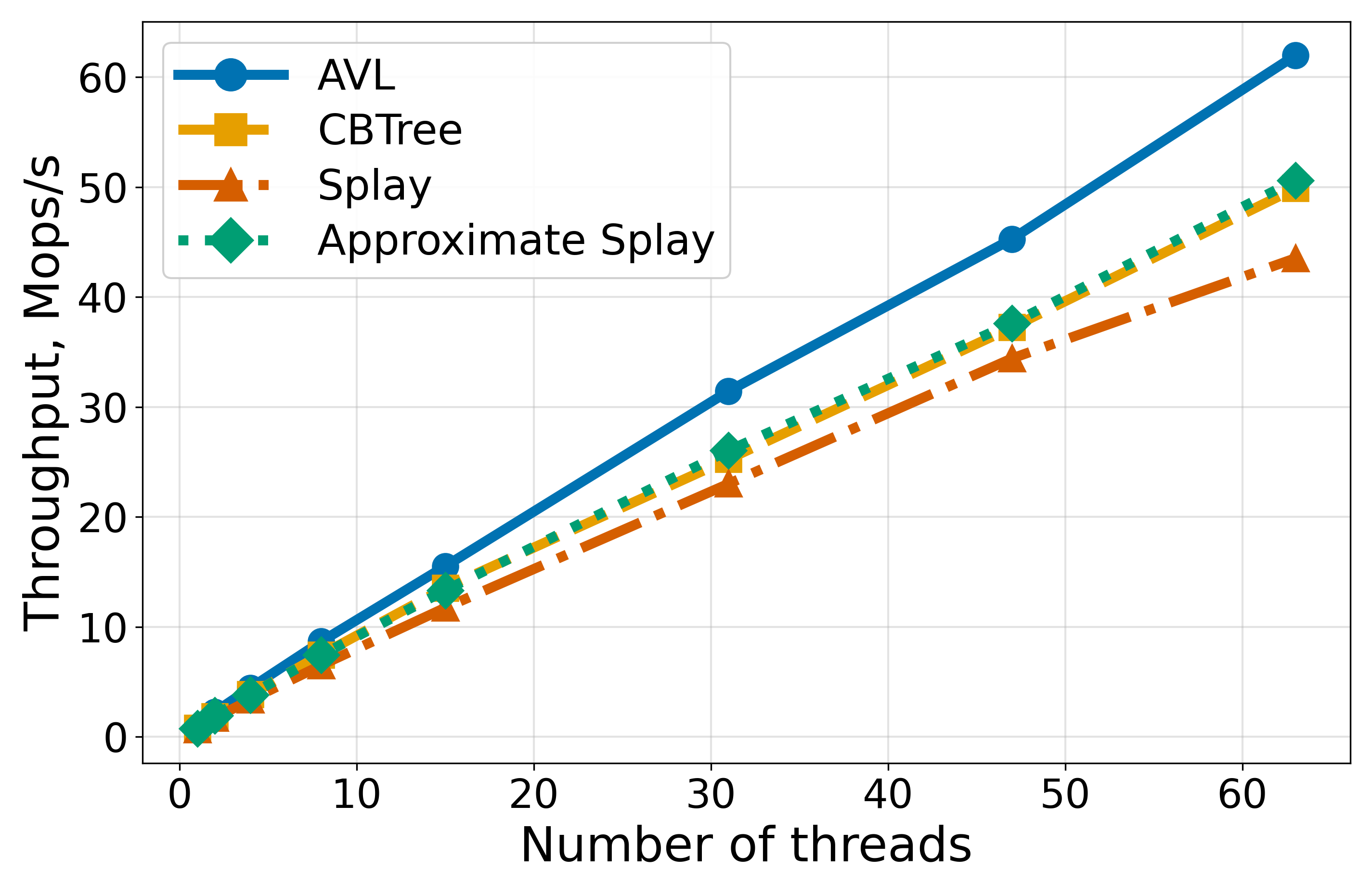}
    \caption{\texttt{uniform}}
\end{subfigure}
\hfill
\begin{subfigure}[t]{0.32\textwidth}
    \centering
    \includegraphics[width=\linewidth]{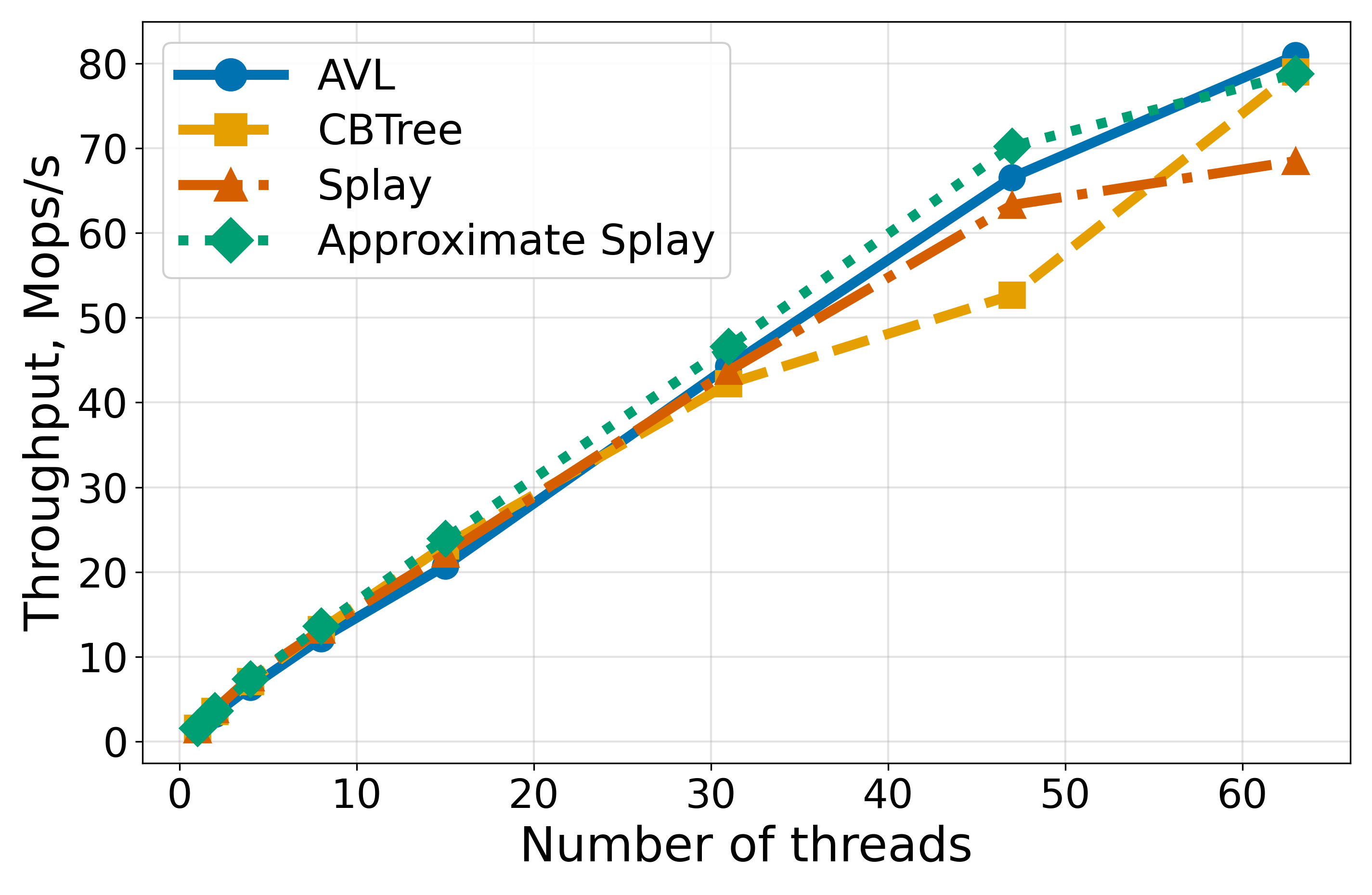}
    \caption{\texttt{zipfian}}
\end{subfigure}
\hfill
\begin{subfigure}[t]{0.32\textwidth}
    \centering
    \includegraphics[width=\linewidth]{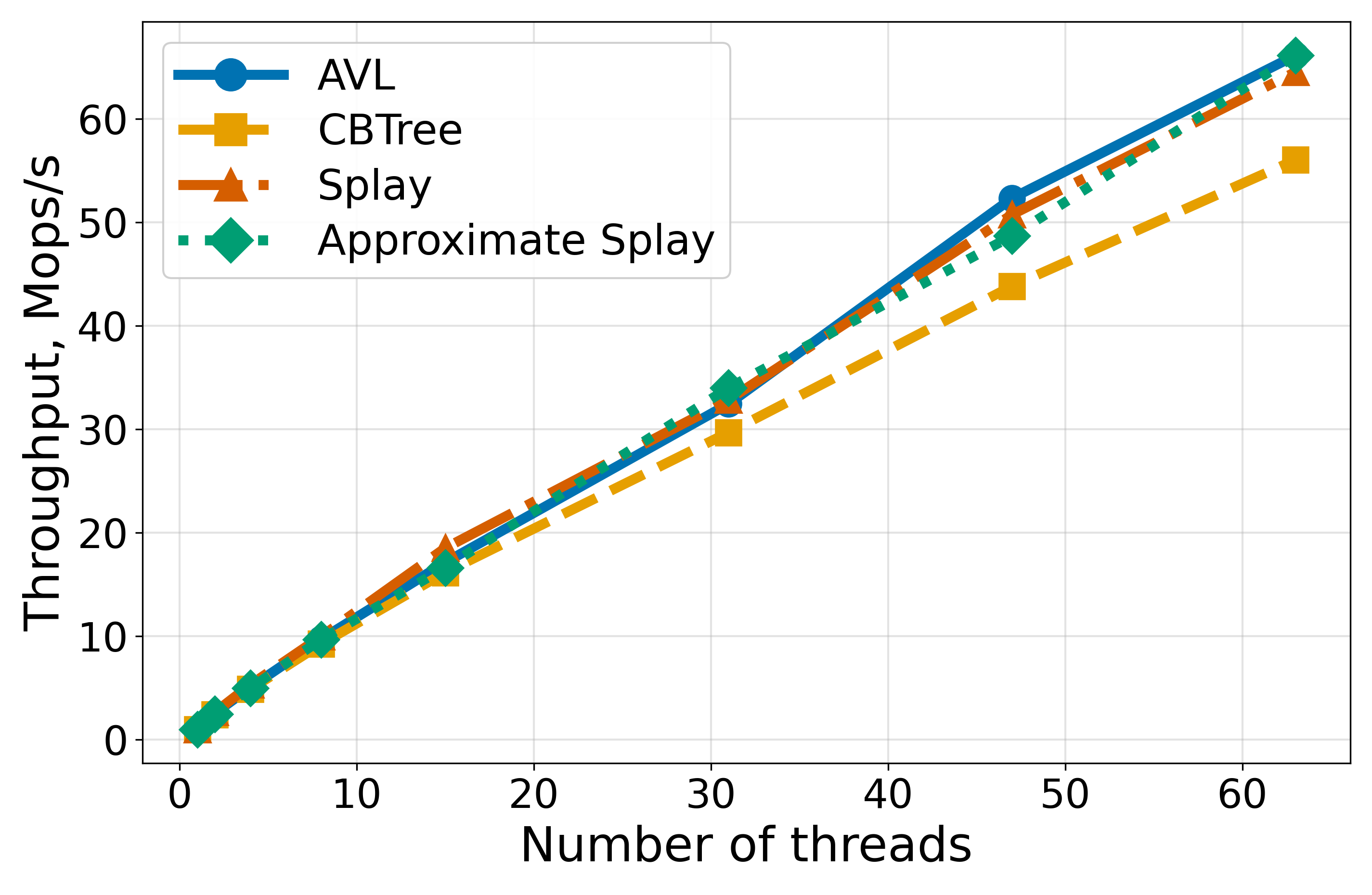}
    \caption{\texttt{90/10}}
\end{subfigure}

\vspace{0.2em}

\begin{subfigure}[t]{0.32\textwidth}
    \centering
    \includegraphics[width=\linewidth]{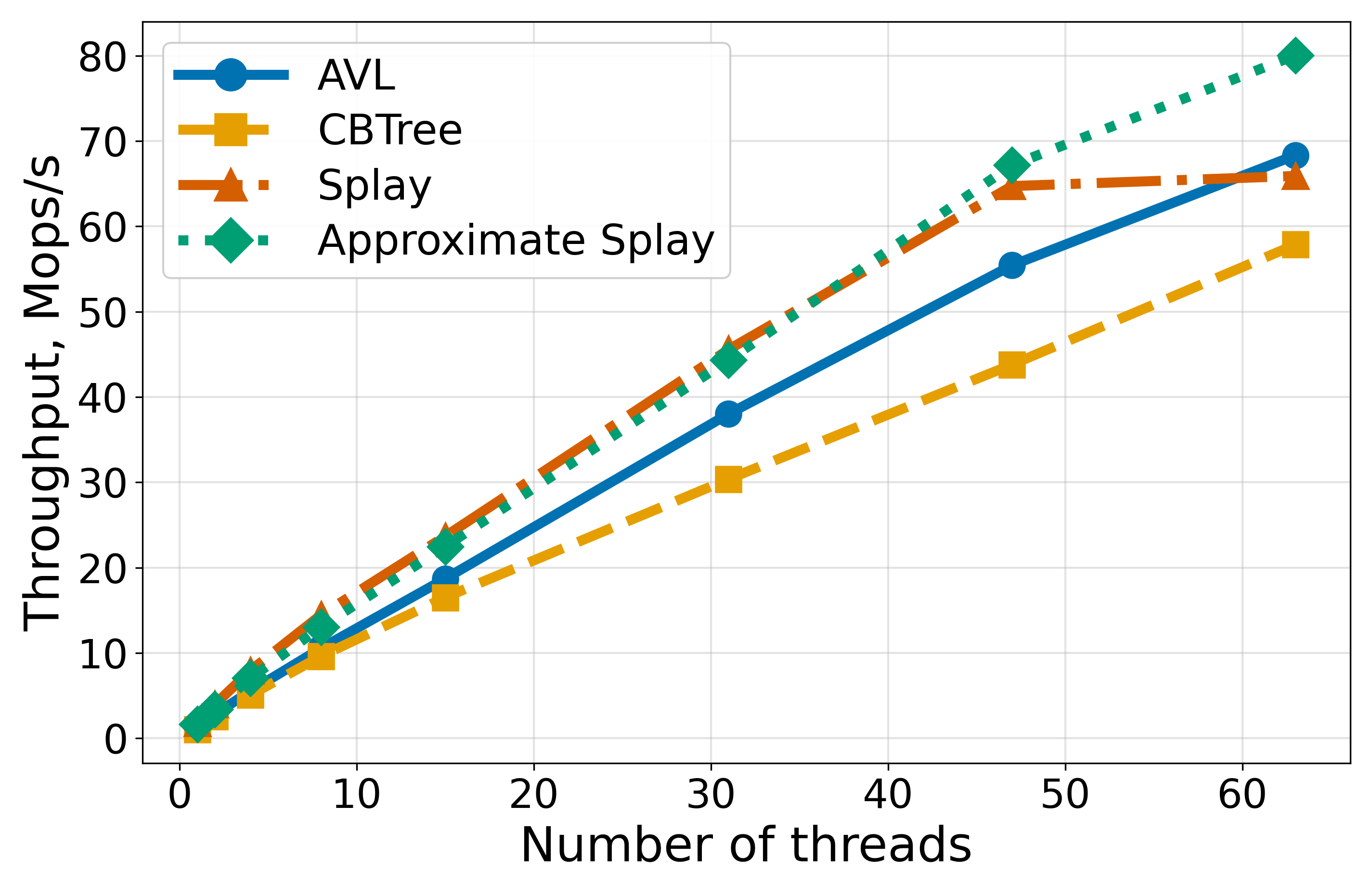}
    \caption{\texttt{95/5}}
\end{subfigure}
\hspace{0.08\textwidth}
\begin{subfigure}[t]{0.32\textwidth}
    \centering
    \includegraphics[width=\linewidth]{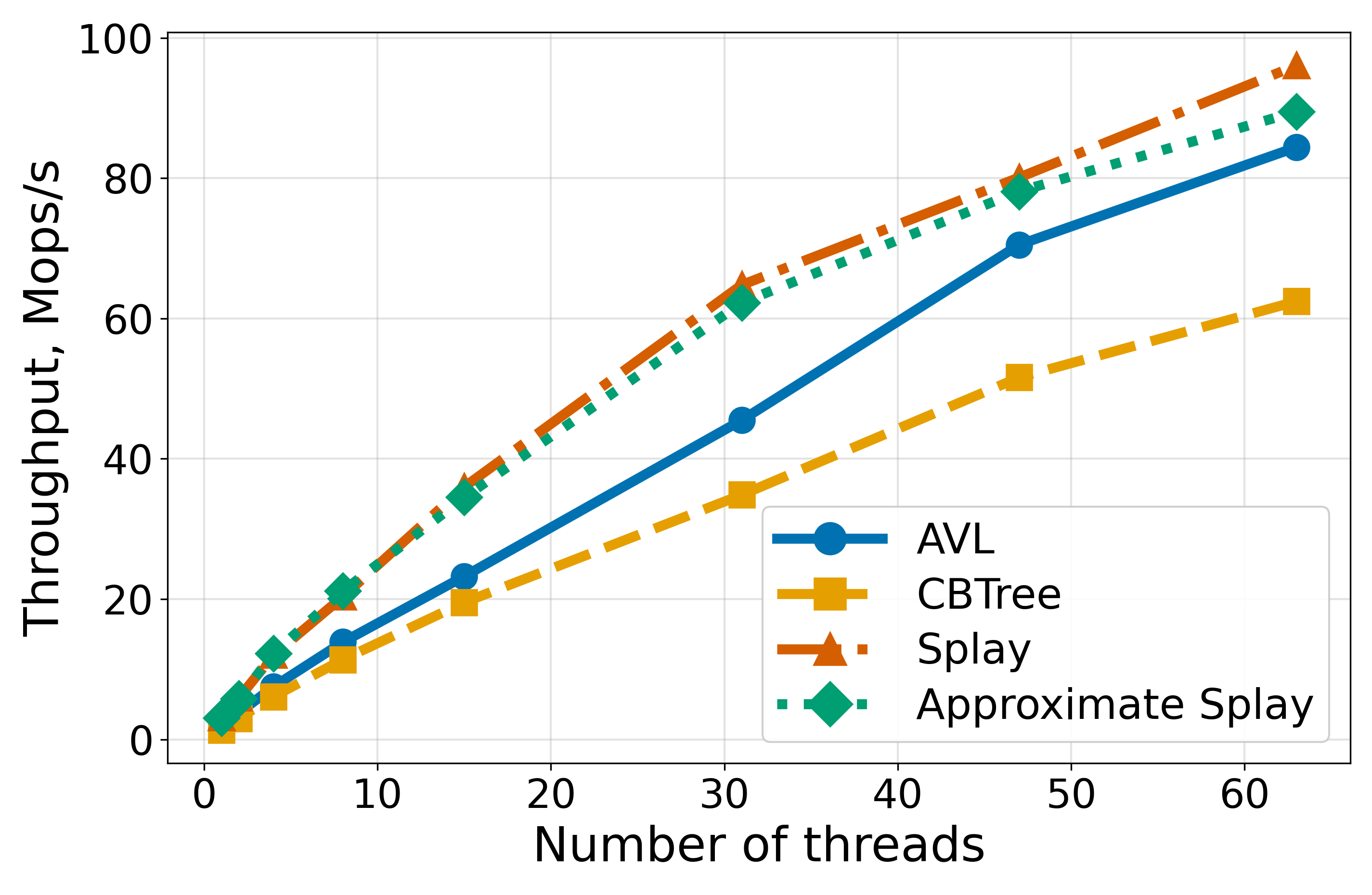}
    \caption{\texttt{99/1}}
\end{subfigure}
\vspace{-0.3cm}
\caption{Throughput of the implementations on read-only workloads with range size $10^6$.}
\label{fig:read-throughput}
\vspace{-0.3cm}
\end{figure}

\begin{figure}[!t]
\centering
\begin{subfigure}[t]{0.32\textwidth}
    \centering
    \includegraphics[width=\linewidth]{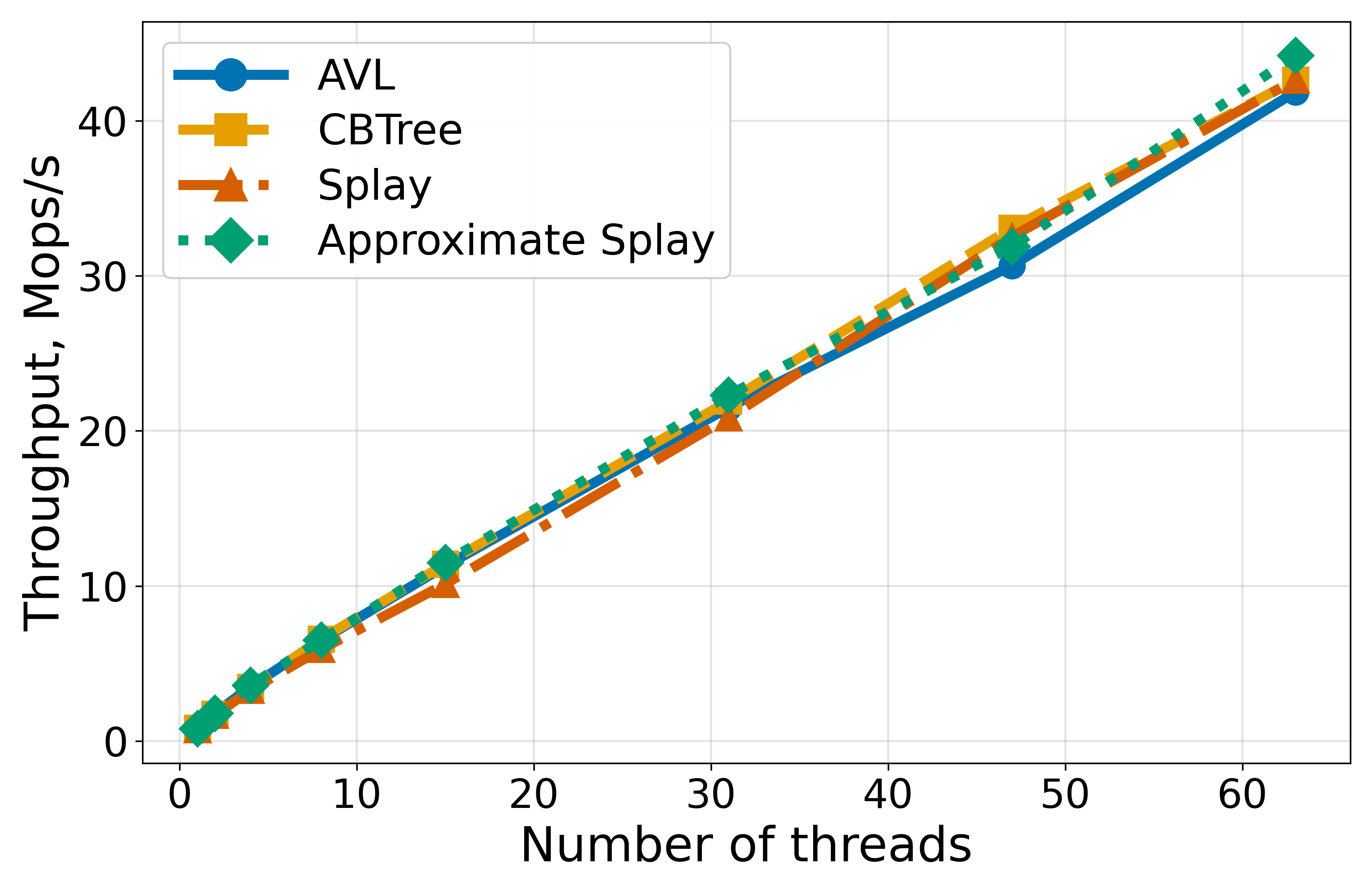}
    \caption{\texttt{uniform}}
\end{subfigure}
\hfill
\begin{subfigure}[t]{0.32\textwidth}
    \centering
    \includegraphics[width=\linewidth]{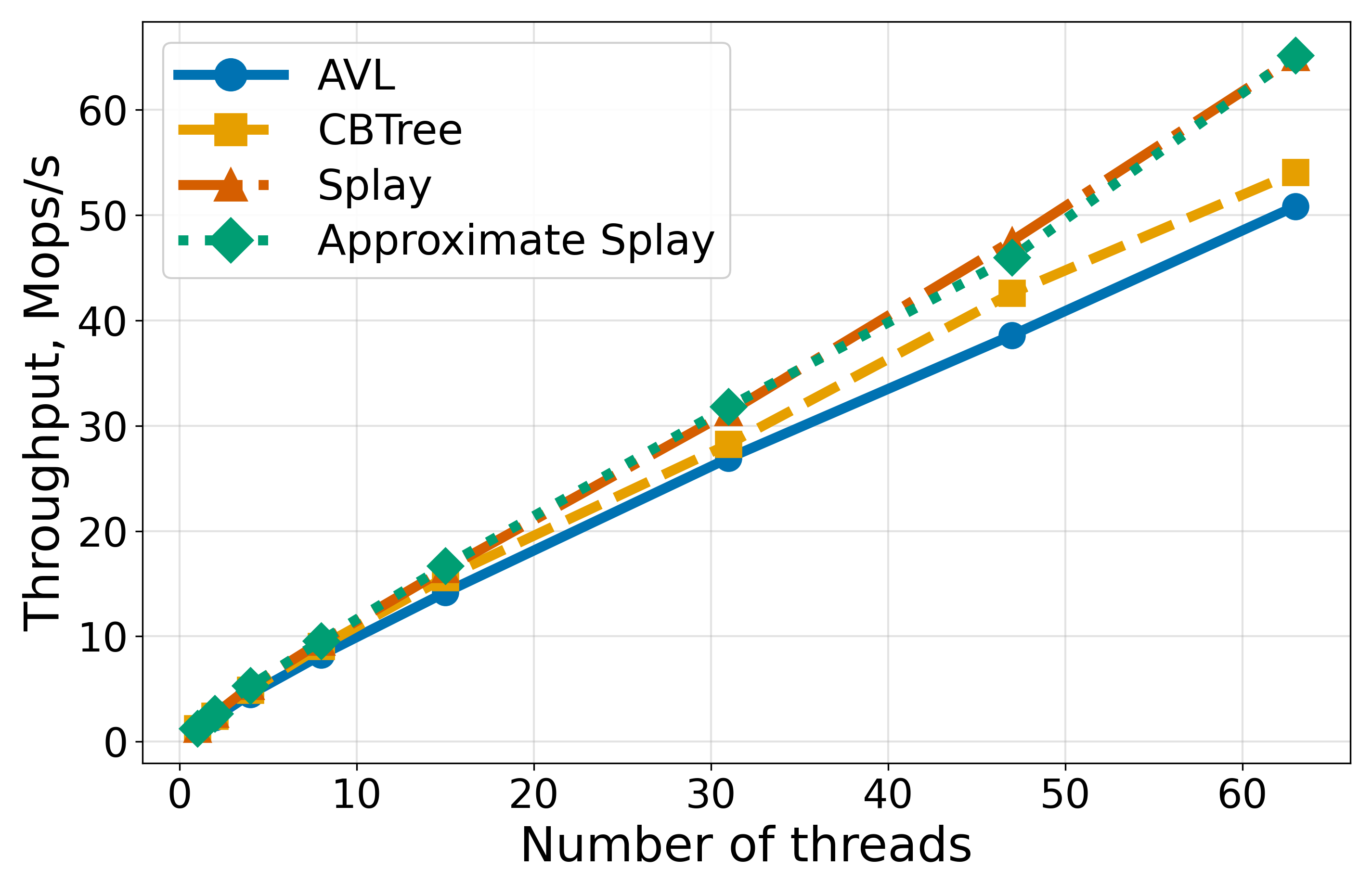}
    \caption{\texttt{zipfian}}
\end{subfigure}
\hfill
\begin{subfigure}[t]{0.32\textwidth}
    \centering
    \includegraphics[width=\linewidth]{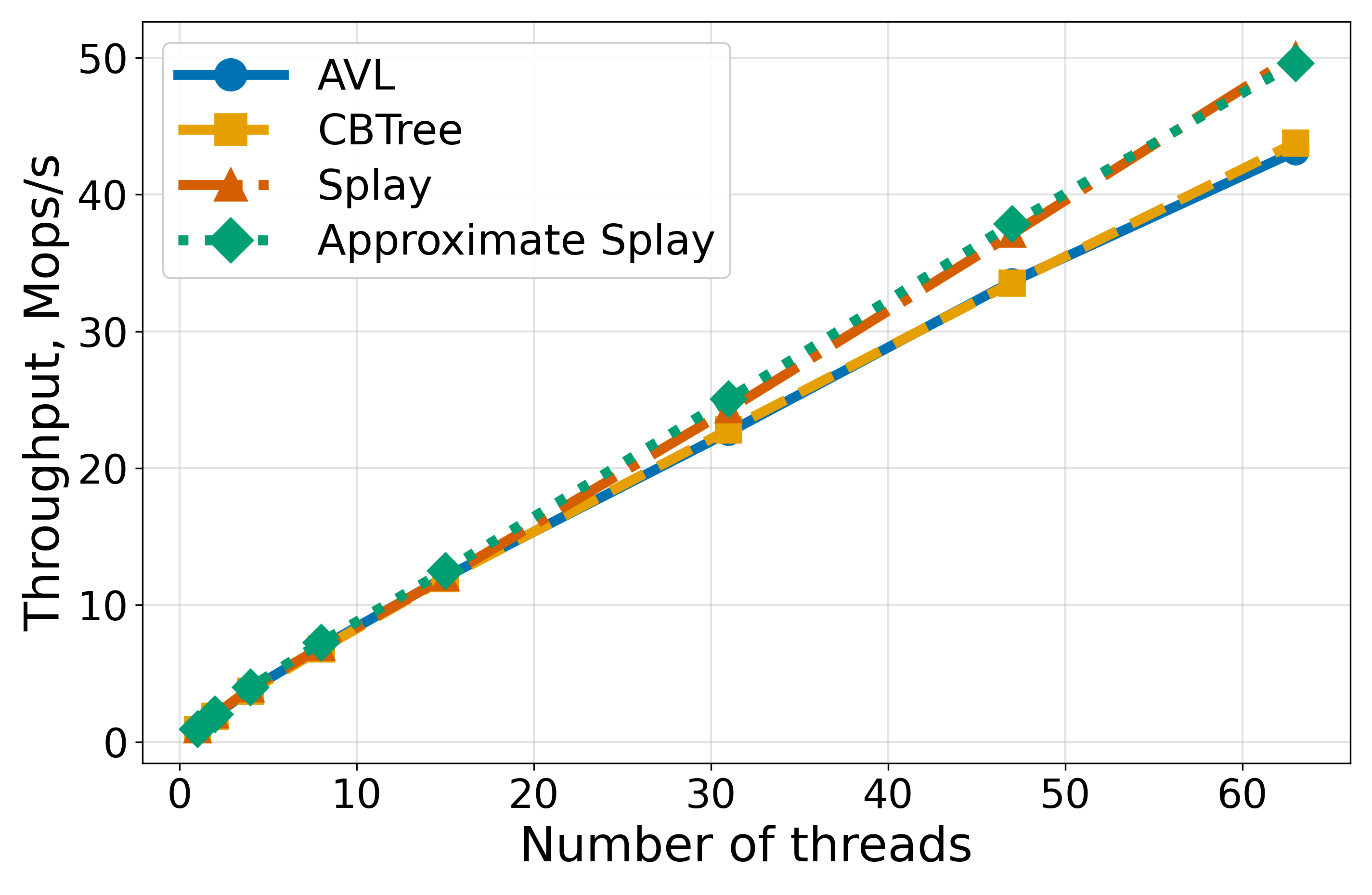}
    \caption{\texttt{90/10}}
\end{subfigure}

\vspace{0.2em}

\begin{subfigure}[t]{0.32\textwidth}
    \centering
    \includegraphics[width=\linewidth]{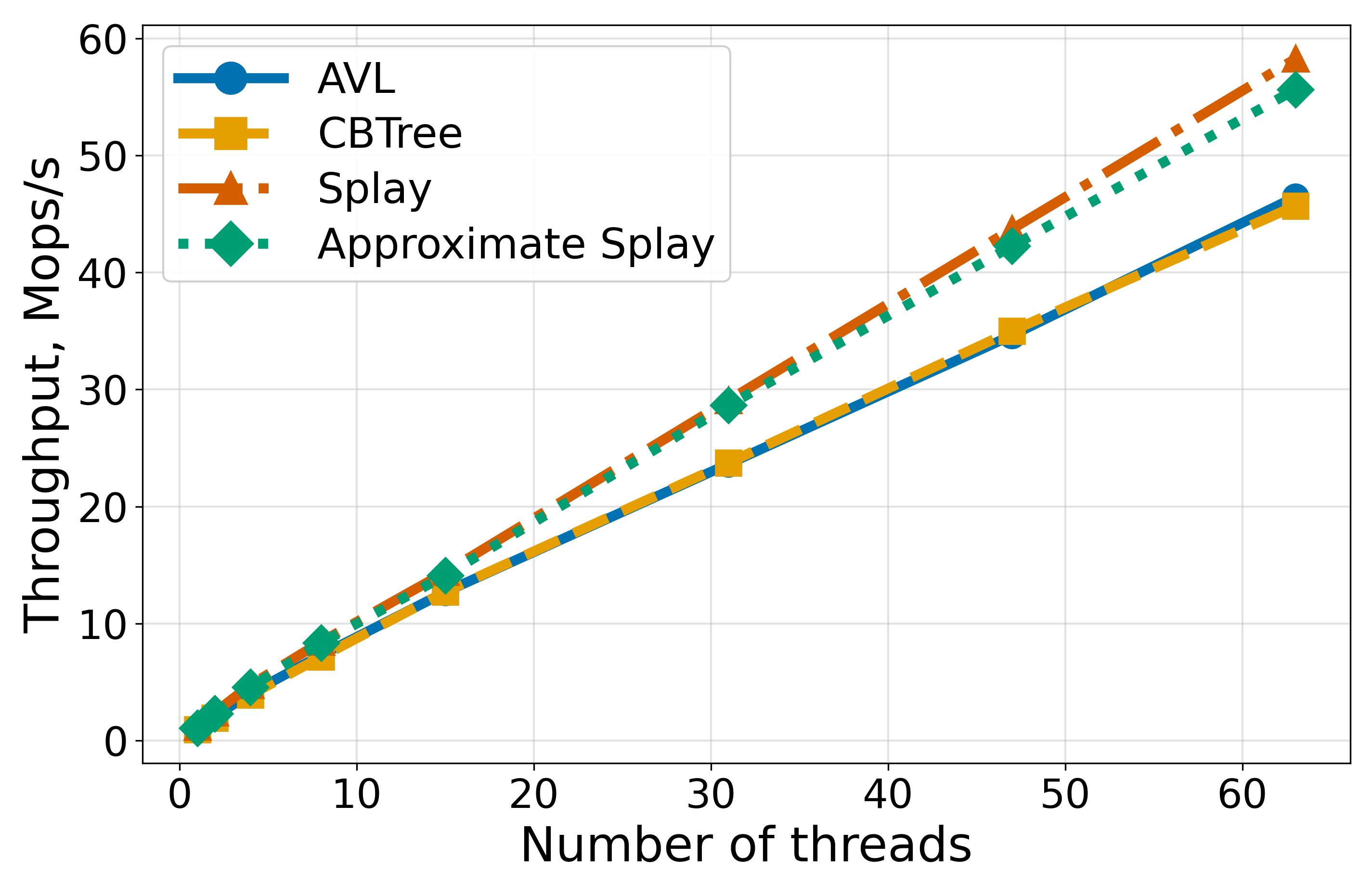}
    \caption{\texttt{95/5}}
\end{subfigure}
\hspace{0.08\textwidth}
\begin{subfigure}[t]{0.32\textwidth}
    \centering
    \includegraphics[width=\linewidth]{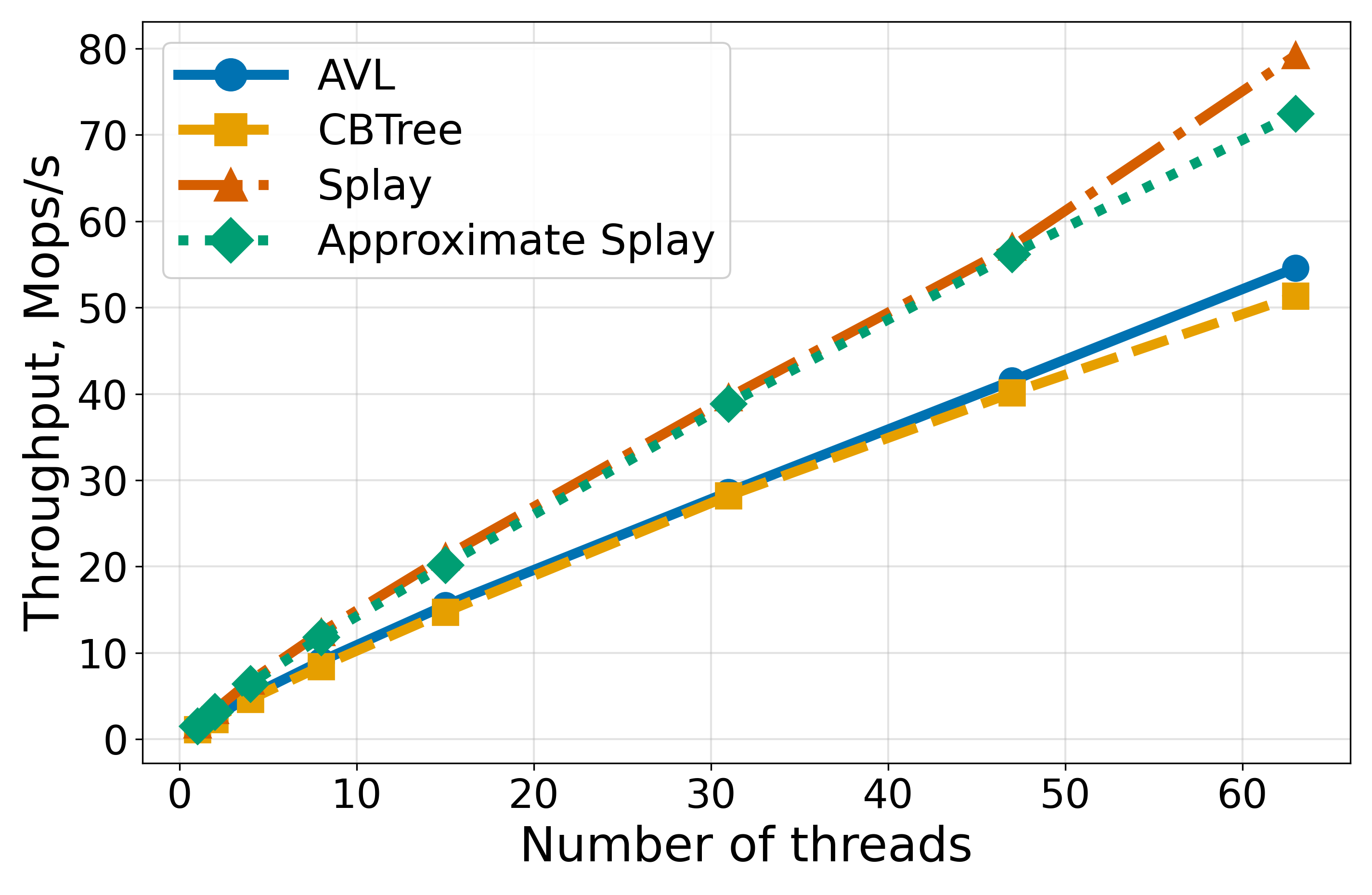}
    \caption{\texttt{99/1}}
\end{subfigure}
\vspace{-0.3cm}
\caption{Throughput of the implementations on update workloads.}
\label{fig:update-throughput}
\vspace{-0.5cm}
\end{figure}

\newpage
\bibliography{references.bib}

\appendix

\section{Deferred proofs}

\begin{lemma}
\label{lem:height}
The expected amortized time incurred by the $i$-th operation is at most $O(\log(AC_i / ac_i(t)) + \log(m / ac_m(t)))$, where $t$ is the node found by the operation, $ac_i(t)$ is the total number of accesses to $t$, $AC_i$ is the total number of accesses before the $i$-th access, and $m$ is the total number of requests.
\end{lemma}
\begin{proof}
Let us start with the simplest case. If the depth of the node does not exceed $B \cdot \log AC_i / ac_i(t)$ by the condition in Line~\ref{line:slack:1}, the statement follows directly.

Now, we consider the case when the depth exceeds that bound.
Let us calculate the expected cost of the change of the potential and the cost of the operation.

We split the cost of an access into two parts. In the first part, we traverse to node $s$ at depth $A \cdot \log(AC_i / ac_i(t))$. Then, we traverse to the target node $t$ and splay it up to $s$ with some probability. For simplicity, we assume that $t$ takes the place of $s$ and not a slightly higher position due to the parity of the depth.

For the second part, we repeat the proof from \cite{albers2002randomized} and provide the full version in the Appendix.

Let $w(u)$ be some weight function on nodes, let $\sigma(u)$ be the sum of weights of all nodes in the subtree of $u$, and let the rank be $\rho(u) = \log \sigma(u)$.
Suppose that we decide to make a rotation with probability $p$ (\texttt{PTHRESHOLD}).
Thus, we will introduce a potential function $\Phi = (1/p + d) \cdot \sum \rho(u)$ where $d$ is the cost of one rotation.

The expected amortized cost of the second part is $E[T] + E[\Delta \Phi]$ where $T$ is the cost of traversal from $s$ and rotations to $s$ and $\Delta \Phi$ is the change in the potential. Note that the first part of an access up to $s$ does not change the potential and always is $O(\log(AC_i / ac_i(t)))$.

The algorithm performs rotations with probability $p$. Since the expectation of the sum is the sum of the expectations, we can consider each rotation separately. Each rotation consists of three parts: the initial cost of the traversal, the cost of the rotation times the probability $p$, and the change in the potential times the probability $p$. Thus, traversal plus rotation for each node is $1 + p \cdot d$.

Now, we will bound the change in the potential. Suppose we rotate a node $u$. Let $v$ and $w$ be its parent and grandparent, if exists.

\begin{enumerate}
\item \textbf{Zig rotation.}
If the edge is rotated, the change in the potential is $\rho'(u) + \rho'(v) - \rho(u) - \rho(v)$ because the ranks change only for $u$ and $v$.
Thus, the expected amortized cost is $(1 + pd) + p (1/p + d) \cdot (\rho'(u) + \rho'(v) - \rho(u) - \rho(v)) \leq (1 + pd) + (1 + pd) \cdot (\rho'(u) - \rho(u)) \leq C + 3 \cdot C \cdot (\rho'(u) - \rho(u))$. The inequalities hold since $\rho(v) \geq \rho'(v)$.

\item \textbf{Zig-zig and zig-zag rotations.}
There are two rotations, which lead to $2 (1 + pd)$, while the change in the potential is $\rho'(u) + \rho'(v) + \rho'(w) - \rho(u) - \rho(v) - \rho(w)$, because only the ranks of $u$, $v$, and $w$ are changed. Therefore, the expected amortized cost does not exceed $2 (1 + pd) + p (1/p + d) \cdot (\rho'(u) + \rho'(v) + \rho'(w) - \rho(u) - \rho(v) - \rho(w)) \leq (1 + pd) \cdot (2 + \rho'(u) + \rho'(v) + \rho'(w) - \rho(u) - \rho(v) - \rho(w))$. Using the same technique as in~\cite{sleator1985self}, we can show that $\rho'(u) + \rho'(v) + \rho'(w) - \rho(u) - \rho(v) - \rho(w) \leq 3 \cdot (\rho'(u) - \rho(u))$. Hence, the amortized cost of that rotation is $3 \cdot (1 + pd) \cdot (\rho'(u) - \rho(u))$.
\end{enumerate}

Finally, the expected amortized cost of traversal from $s$ to $u$ and rotations equals the sum over all rotations. This sum is bounded by $3 \cdot (1 + pd) \cdot (\rho'(t) - \rho(t)) = 3 \cdot (1 + pd) \cdot (\rho(s) - \rho(t))$, since $t$ takes the place of $s$.

To finish the proof, we fix the weight function $w$ of nodes as the number of accesses to the corresponding key at the end. This gives the required bound $3 \cdot (1 + pd) \cdot \log (m / ac_m(t))$ since $\rho(s) \leq \log m$.
\end{proof}

\begin{theorem}
A tree with such rotations serves accesses in $O(1/p \cdot \sum_x ac(x) \cdot \log (m / ac(x)))$ amortized time, where $m$ is the total number of accesses and $ac(x)$ is the total number of accesses to $x$.
\end{theorem}
\begin{proof}
Choosing the weight function as the number of accesses bounds the initial potential from above by $(1/p + d) \cdot \sum\limits_x \log m$ and bounds the final potential from below by $(1/p + d) \cdot \sum\limits_x \log ac(x)$. The contribution of the potential is the difference between the initial and final potentials and is thus upper-bounded by the desired bound.

Now, we consider all accesses to node $x$. By Lemma~\ref{lem:height}, their total cost is $\sum_{i = 1}^{ac(x)} O(\log (AC_{o_i} / i) + \log (m / ac(x)))$, where $o_i$ is the identifier of the $i$-th operation on $x$. Since $\sum_{i = 1}^{ac(x)} \log i \geq \sum_{i = ac(x) / 2}^{ac(x)} \log i \geq ac(x) / 2 \cdot \log (ac(x) / 2)$ and $AC_{o_i} \leq m$, we get $\sum_{i = 1}^{ac(x)} O(\log (AC_{o_i} / i)) = O(ac(x) \cdot \log (m / ac(x)))$, which is exactly what we require.
\end{proof}

\section{Listings}

\begin{center}
\begin{lstlisting}[multicols=]
if rnd() < 1 / 2^global_counter:
    global_counter += 1

if rnd() < 1 / 2^node.counter:
    node.counter += 1

target = global_counter - node.counter
\end{lstlisting}
\captionof{lstlisting}{Approximate computation of the target depth}
\label{alg:approximate}
\end{center}

\end{document}